\documentclass[12pt, reqno]{amsart} 
\usepackage{amsmath, amsthm, amscd, amsfonts, amssymb, graphicx, color}
\usepackage[bookmarksnumbered, colorlinks, plainpages]{hyperref}

\newtheorem{theorem}{Theorem}[section]
\newtheorem{lemma}[theorem]{Lemma}
\newtheorem{proposition}[theorem]{Proposition}
\newtheorem{corollary}[theorem]{Corollary}
\theoremstyle{definition}

\newtheorem{example}[theorem]{Example}

\theoremstyle{remark}

\numberwithin{equation}{section}

\begin{document}
\title[$b$-symbol distance distribution of repeated-root cyclic codes]
{$b$-symbol distance distribution of repeated-root cyclic codes}
\author[Mostafanasab and Sengelen Sevim]{Hojjat Mostafanasab and Esra Sengelen Sevim}

\date{}
\subjclass[2010]{}
\keywords{$b$-symbol-pair, distance distribution, cyclic codes.}

\begin{abstract}
Symbol-pair codes, introduced by Cassuto and Blaum \cite{CB}, have been raised for symbol-pair
read channels. This new idea is motivated
by the limitations of the reading process in high-density data storage technologies. 
Yaakobi et al. \cite{YBS} introduced codes for $b$-symbol read channels, where the read operation is
performed as a consecutive sequence of $b>2$ symbols.
In this paper, we come up with a method to compute the $b$-symbol-pair distance of two  $n$-tuples, where $n$ is a positive
integer. Also, we deal with the $b$-symbol-pair distances of some kind of cyclic codes of length $p^e$ over $\mathbb{F}_{p^m}$. 
\end{abstract}

\maketitle

\section{Introduction}

Recently, it is possible to write information on storage devices with high resolution using advances in data storage systems. However, it causes a problem of the gap between write resolution and read resolution. 
 Cassuto and Blaum \cite{CB,CB2} laid out a framework
for combating pair-errors, relating pair-error correction
capability to a new metric called pair-distance. They proposed the model of symbol-pair read channels. Such channels are mainly motivated by magnetic-storage channels with high write resolution, due to physical limitations, each channel contains contributions from two adjacent symbols. Cassuto and Listsyn \cite{CL} studied algebraic construction of cyclic symbol-pair codes. Yaakobi et al. \cite{YBS1} proposed efficient decoding algorithms for the cyclic symbol-pair codes. Chee et al.
\cite{CKW,CJK} established a Singleton-type bound for symbol-pair codes and constructed codes that meet the Singleton-type
bound. Hirotomo et al. \cite{HTM} proposed the decoding algorithm for symbol-pair codes based on the newly defined
parity-check matrix and syndromes.\\
For this new channels, the codes defined as usual over some discrete symbol alphabet, but whose reading from the channel is performed as overlapping pairs of symbols. 
Let $\Xi$ be the alphabet consisting of $q$ elements. Each element in $\Xi$ is called a symbol. We use $\Xi^n$ to denote the set of all 
$n$-tuples, where $n$ is a positive integer. 
In the symbol-pair read channel, there are in fact two channels. If the stored information is ${x}=(x_0,x_1,\dots,x_{n-1})\in \Xi^n$, then the symbol-pair read vector of $x$ is
$$
\pi({x})=[(x_0,x_1),(x_1,x_2),\dots,(x_{n-2},x_{n-1}),(x_{n-1},x_0)],
$$
and the goal is to correct a large number of the so called symbol-pair errors. The
pair distance, $d_p(x, y)$, between two pair-read vectors
$x$ and $y$ is the Hamming distance over the symbol-pair
alphabet $(\Xi \times \Xi)$ between their respective pair-read vectors, that is, $d_p(x,y)=d_H(\pi(x), \pi(y))$. 
The minimum pair distance of a code $\mathcal{C}$ is defined as
$d_p(\mathcal{C})={\rm min}\{d_p(x,y)|x,y\in \mathcal{C}\mbox{ and }x\neq y\}$.
Accordingly, the pair weight of $x$ is $\omega_p(x)= \omega_H(\pi(x))$.
If $\mathcal{C}$ is a linear code, then the minimum pair-distance of $\mathcal{C}$ is the smallest pair-weight of nonzero codewords of $\mathcal{C}$. The minimum pair-distance is one of the important parameters of symbol-pair codes. This distance distribution is very difficult to compute in general, however, for the class of cyclic codes of length $p^e$ over $\mathbb{F}_{p^m}$, their Hamming distance has been completely determined in \cite{HQD}. In \cite{ZSW}, Zhu et al. investigated the symbol-pair
distances of cyclic codes of length $p^e$ over $\mathbb{F}_{p^m}$. 
 
For $b\geq 3$, the $b$-symbol read vector corresponding to the
vector $x = (x_0, x_1, \dots , x_{n-1})\in \Xi^n$ is defined as
$$\pi_b(x)= [(x_0,x_1,\dots , x_{b-1}),(x_1,x_2,\dots , x_{b}),\dots ,(x_{n-1}, x_0, \dots , x_{b-2})]\in (\Xi^b)^n.$$
We refer to the elements of $\pi_b(x)$ as $b$-symbols. The $b$-symbol
distance between $x$ and $y$, denoted by $d_b(x, y)$, is defined as
$d_b(x, y) = d_H(\pi_b(x), \pi_b(y))$.
Similarly, we define the $b$-weight of the vector $x$ as
$\omega_H(\pi_b(x))$.
In the analogy of the definition of symbol-pair codes, 
the minimum $b$-symbol distance
of $\mathcal{C}$, $d_b(\mathcal{C})$, is given by $d_b(\mathcal{C})={\rm min}\{d_b({x},{y})|{x},{y}\in \mathcal{C}\mbox{ and }{x}\neq {y}\}$. For more information on these notions see \cite{YBS}. 

We can rewrite \cite[Proposition 9]{YBS} for any arbitrary alphabet $\Xi$.
\begin{proposition}\label{propos}
Let $x\in\Xi^n$ be such that $0<\omega_H(x)\leq n-(b-1)$. Then
$$\omega_H(\mathcal{C})+b-1\leq \omega_b(\mathcal{C})\leq b\cdot \omega_H(\mathcal{C}).$$
\end{proposition}

 Referring to Proposition \ref{propos}, we see that:
\begin{corollary}\label{lemm}
Let $\mathcal{C}$ be a code. If $0<d_H(\mathcal{C})\leq n-(b-1)$, then
$$d_H(\mathcal{C})+b-1\leq d_b(\mathcal{C})\leq b\cdot d_H(\mathcal{C}).$$
\end{corollary}

In the next section we give a method to calculate the $b$-symbol distance of two $n$-tuples.
We know that all cyclic codes of length $p^e$ over a finite field of characteristic $p$ are generated by a
single ``monomial'' of the form $(x-1)^i$, where $0\leq i\leq p^e$ (see \cite{HQD}). Determining
the $b$-symbol-pair distances of some kind of these cyclic codes is the main purpose of the next section. 

\section{Main results}

In the following theorem we give a formula to calculate the $b$-symbol distance of two $n$-tuples.
\begin{theorem}\label{main1}
Let $x=(x_1,x_2,\dots,x_n)$ and $y=(y_1,y_2,\dots,y_n)$ be two vectors in $\Xi^n$ with $0<d_H(x,y)\leq n-(b-1)$. Suppose that 

$$A=\{1,2,\dots,n\}\backslash\{r,r+1,r+2,\dots,s\mid r,s \mbox{ are}
\mbox{ such that }s-r\geq b-2 \mbox{ and } x_i=y_i$$ 
$$\mbox{ for each }  r\leq i\leq s \mbox{ and indices may wrap around modulo~} n\},$$
and $A=\cup_{l=1}^{L}B_l$ is a minimal partition of the set $A$ to subsets of consecutive
indices $($every subset $B_l=[s_l,e_l]$ is the sequence of all indices
between $s_l$ and $e_l$, inclusive, and is the smallest integer that
achieves such partition, also indices may wrap around modulo $n$$)$. Then
$$d_b(x,y)=d_H(x,y)+e+L(b-1),$$
where $e=|\{i\mid i\in B_l\mbox{ for some } 1\leq l\leq L \mbox{ such that } x_i=y_i\}|$.
\end{theorem}
\begin{proof}
Since the partition is minimal, there are no two indices $i,i+j$, where $j\in\{1,\dots,b-1\}$, that belong
to different subsets $B_l,B_{l^{\prime}}$.
The $b$-symbol distance between $x$ and $y$ is equal to the sum of the sizes of the $b$-tuple subsets
{\scriptsize
$$\{(s_l-b+1,s_l-b+2,\dots,s_l),(s_{l}-b+2,s_l-b+3,\dots,s_l,s_l+1),\dots,(s_l,s_l+1,\dots,s_l+b-1),$$
$$(s_l+1,s_l+2,\dots,s_l+b),\dots,(e_l,e_l+1,\dots,e_l+b-1)\}.$$}
The number of $b$-tuples in each $b$-tuple subset equals $|B_l|+b-1$, whence
$d_b(x,y)=\sum_{l=1}^{L}B_l+L(b-1)$. Furthermore, it is easy to see that $\sum_{l=1}^{L}B_l=d_H(x,y)+e$
where $e=|\{i\mid i\in B_l\mbox{ for some } 1\leq l\leq L \mbox{ such that } x_i=y_i\}|$.
\end{proof}

\begin{corollary}\label{maincoro}
Let $x=(x_1,x_2,\dots,x_n)\in\Xi^n$ with $0<\omega_H(x)\leq n-(b-1)$. Suppose that
$$A=\{1,2,\dots,n\}\backslash\{r,r+1,r+2,\dots,s\mid r,s \mbox{ are}
\mbox{ such that }s-r\geq b-2 \mbox{ and } x_i=0$$ 
$$\mbox{ for each }  r\leq i\leq s \mbox{ and indices may wrap around modulo~} n\},$$
and $A=\cup_{l=1}^{L}B_l$ is a minimal partition of the set $A$ to subsets of consecutive
indices $($every subset $B_l=[s_l,e_l]$ is the sequence of all indices
between $s_l$ and $e_l$, inclusive, and is the smallest integer that
achieves such partition, also indices may wrap around modulo $n$$)$. Then $\omega_b(x)=\omega_H(x)+e+L(b-1),$
where $$e=|\{i\mid i\in B_l\mbox{ for some } 1\leq l\leq L \mbox{ such that } x_i=0\}|.$$
\end{corollary}

\begin{example}
Let $n=15$, $b=4$ and $x=(0,0,1,3,0,5,0,0,0,2,0,7,0,0,0)\in\mathbb{Z}^{15}$.
We list all of the $4$-tuples as follows:
$$(0,0,1,3),(0,1,3,0),(1,3,0,5),(3,0,5,0),(0,5,0,0),(5,0,0,0),(0,0,0,2),$$
$$(0,0,2,0),(0,2,0,7),(2,0,7,0),(0,7,0,0),(7,0,0,0),(0,0,0,0),(0,0,0,1).$$
Hence $\omega_4(x)=13$. On the other hand, $\omega_H(x)=5$, $e=2$ and $L=2$. 
Therefore, the equation $\omega_b(x)=\omega_H(x)+e+L(b-1)$ holds.
\end{example}

\begin{theorem}\label{mainthm} \rm{(\cite[Theorem 6.4]{HQD})} 
Let $\mathcal{C}$ be a cyclic code of length $p^e$ over $\mathbb{F}_{p^m}$. 
Then $\mathcal{C}=\langle (x-1)^i\rangle\subseteq \frac {\mathbb{F}_{p^m}[x]}{\langle x^{p^e}-1\rangle}$, for $i\in\{0,1,\dots,p^e\}$. 
The Hamming distance $d_H(\mathcal{C})$ is determined by
{\small
\begin{equation*}d_H(\mathcal{C})=
\begin{cases}1& if~i=0,\\
\beta+2 & if~\beta p^{e-1}+1\leq i\leq (\beta +1)p^{e-1}~where~0\leq \beta \leq p-2,\\
(t+1)p^k& if~p^e-p^{e-k}+(t-1)p^{e-k-1}+1\leq i \leq p^e-p^{e-k}+tp^{e-k-1},\\
& where~1\leq t \leq p-1,~and~1\leq k \leq e-1,\\
0 & if~i=p^e.
\end{cases}
\end{equation*}}
\end{theorem}

From now on, in order to simplify the notation, for $i\in\{0,1,\dots,p^e\}$, we denote each code
$\langle (x-1)^i\rangle$ by $\mathcal{C}_i$.

\begin{proposition}\label{firstthm}
If $b\leq p^e$, then $d_b(\mathcal{C}_{0})=b$.
\end{proposition}
\begin{proof}
By Theorem \ref{mainthm}, we have that $d_H(\mathcal{C}_{0})=1$. 
So, by Corollary \ref{lemm}, $b\geq d_b(\mathcal{C}_{0})\geq d_H(\mathcal{C}_{0})+b-1=b$. Hence $d_b(\mathcal{C}_{0})= b$.
\end{proof}

\begin{proposition}\label{prop1}
Let $b<p^e$. Then
$b+1\leq d_b(\mathcal{C}_i)\leq 2b$ for every $1\leq i\leq p^{e-1}$.
\end{proposition}
\begin{proof}
By Theorem \ref{mainthm}, $d_H(\mathcal{C}_i)=2$ for every $1\leq i\leq p^{e-1}$. Hence,
$2b\geq d_b(\mathcal{C}_i)\geq 2+(b-1)=b+1$, by Corollary \ref{lemm}.
\end{proof}

Notice that, for two codes $\mathcal{C}, \mathcal{C}^\prime\subseteq \mathbb{F}_{p^m}^{p^e}$ with $\mathcal{C} \subseteq \mathcal{C}^\prime$, we have $d_b(\mathcal{C})\geq d_b(\mathcal{C}^\prime)$. 
We define $d_b(\mathcal{C}_{p^e})=0$.


\begin{proposition}\label{m1}
Let $b\leq p$ and $e=1$.
Then $d_b(\mathcal{C}_i)=
i+b$ for each $0\leq i \leq p-b$.
\end{proposition}
\begin{proof}
By Theorem \ref{mainthm}, $d_H(\mathcal{C}_i)=i+1$ for $0\leq i \leq p-1$. Assume that $0\leq i \leq p-b$.
Hence, by Corollary \ref{lemm}, $d_b(\mathcal{C}_i)\geq i+1+b-1=i+b$. Moreover
$\omega_b((x-1)^i)=i+1+(b-1)=i+b$. Then $d_b(\mathcal{C}_i)=i+b$. 
\end{proof}
 


 \begin{theorem}\label{m3}
Let $e\geq2$ and $1\leq i\leq p^{e-1}$ such that $i+b\leq p^e$ and $i\leq b$. 
Then $d_b(\mathcal{C}_i)=i+b$.
\end{theorem}
\begin{proof}
Since $i+b\leq p^e$, then by Corollary \ref{maincoro}, $\omega_b((x-1)^i)=i+1+(b-1)=i+b$.
So, $d_b(\mathcal{C}_i)\leq i+b$.
By Proposition \ref{prop1}, $d_b(\mathcal{C}_i)\geq b+1$.
Let $c(x)$ be a polynomial in $\mathbb{F}_{p^m}[x]$.
If $\omega_b(c(x))=j+b$ for some $1\leq j\leq i-1$, then $i\leq b$ implies
that $c(x)=x^t(a_0+a_1x+\dots+a_jx^j)$ where $a_l$'s are in $\mathbb{F}_{p^m}$,~ $a_0,a_j\neq0$ and $t$ is a non-negative integer.
However $c(x)\notin\mathcal{C}_i$. So $d_b(\mathcal{C}_i)= i+b$.
\end{proof}

\begin{lemma}\label{basiclemma3}
Let $e$ and $k$ be two integers such that $e\geq2$ and $1\leq k\leq e-1$. 
Suppose that $c(x)=(x-1)^{p^e-p^{e-k}}g(x)$ where $g(x)$ is a nonzero polynomial in $\mathbb{F}_{p^m}[x]$ 
with $d:=deg(g(x))< p^{e-k}$ and  $b\leq p^e-d$.
Then 
\begin{enumerate}
\item If $d\leq p^{e-k}-b$ or $g_k= 0$ for every $0\leq k\leq b-p^{e-k}+d-1$, then $\omega_b(c(x))=p^k\omega_b(g(x))$.
\item If $d> p^{e-k}-b$ and $g_k\neq 0$ for some $0\leq k\leq b-p^{e-k}+d-1$, then $\omega_b(c(x))=p^k\big(\omega_b(g(x))-(b-1)+\zeta\big)$
where $\zeta=p^{e-k}-d-1$.
\end{enumerate}
\end{lemma}
\begin{proof}
Assume that $g(x)=\sum_{j=0}^dg_jx^j$. Thus 
$$c(x)=\sum_{i=0}^{p^k-1}x^{ip^{e-k}}g(x)=\sum_{i=0}^{p^k-1}\sum_{j=0}^{d}g_jx^{ip^{e-k}+j}.$$
As usual, we identify  the polynomial $h(x)=h_0+h_1x+\dots+h_nx^n$ with the vector $h=(h_0,h_1,\dots,h_n)$. 
Therefore, we have
$c=(\overbrace{\widehat{g},\dots,\widehat{g}}^{p^k-\mbox{time}})$
where 
$$\widehat{g}=(g_0,\dots,g_d,\overbrace{0,\dots,0}^{(p^{e-k}-d-1)-\mbox{time}}).$$
 We denote $\omega_b(\widehat{g}(x)):=\omega_b(\widehat{g})$. Since
$\pi_b(c)=[\overbrace{\pi_b(\widehat{g}),\dots,\pi_b(\widehat{g})}^{p^k-\mbox{time}}],$
then $\omega_b(c(x))=p^k\omega_b(\widehat{g}(x))$.
On the other hand, $\omega_b(g(x))=\omega_b(g)$, where 
$$g=(g_0,g_1,\dots, g_d,\overbrace{0,\dots,0}^{(p^e-d-1)-\mbox{time}}).$$
We can check that:\\
(1) If $d\leq p^{e-k}-b$ or $g_k= 0$ for every $0\leq k\leq b-p^{e-k}+d-1$, then $\omega_b(g)=\omega_b(\widehat{g})$, i.e., $\omega_b(g(x))=\omega_b(\widehat{g}(x))$. Hence $\omega_b(c(x))=p^k\omega_b(g(x))$.\\
(2) If $d> p^{e-k}-b$ and $g_k\neq 0$ for some $0\leq k\leq b-p^{e-k}+d-1$, then 
$\omega_b(g)=\omega_b(\widehat{g})+(b-1)-\zeta$ where
$\zeta=p^{e-k}-d-1$, i.e., $\omega_b(g(x))=\omega_b(\widehat{g}(x))+(b-1)-\zeta$. So, $\omega_b(c(x))=p^k\big(\omega_b(g(x))-(b-1)+\zeta\big)$.
\end{proof}

\begin{theorem}
Let $e$ and $k$ be two integers such that $e\geq2$ and $1\leq k\leq e-1$. 
If $0\leq i\leq p^{e-k-1}$ such that $b+i\leq p^{e-k}$ and $i\leq b$,
then $d_b(\mathcal{C}_{p^e-p^{e-k}+i})=p^k(b+i)$.
\end{theorem}
\begin{proof}
Fix $0\leq i\leq p^{e-k-1}$ such that $b+i\leq p^{e-k}$ and $i\leq b$. Let $0\neq c(x)\in\mathcal{C}_{p^e-p^{e-k}+i}$. Then, there exists $0\neq f(x)\in\mathbb{F}_{p^m}[x]$ such that $c(x)=(x-1)^{p^e-p^{e-k}}(x-1)^if(x)$. 
Set $g(x):=(x-1)^if(x)$ and $d:=deg(g(x))$.
Without loss of the generality we may assume that $d< p^{e-k}$.
Notice that by Theorem \ref{mainthm}, $\omega_H(g(x))\geq 2$, and by Theorem \ref{m3}, $\omega_b(g(x))\geq b+i$.
Regarding Lemma \ref{basiclemma3}, we consider the following cases:\\
{\bf Case 1.} If $d\leq p^{e-k}-b$ or $g_k= 0$ for every $0\leq k\leq b-p^{e-k}+d-1$, then 
$\omega_b(c(x))=p^k\omega_b(g(x))\geq p^k(b+i)$.\\
{\bf Case 2.} If $d> p^{e-k}-b$ and $g_k\neq 0$ for some $0\leq k\leq b-p^{e-k}+d-1$, then $\omega_b(c(x))=p^k\big(\omega_b(g(x))-(b-1)+\zeta\big)$
where $\zeta=p^{e-k}-d-1$.
If $\omega_H(g(x))\geq b+i$, then Corollary \ref{lemm} implies that $\omega_b(g(x))\geq b+i+b-1$. 
 Hence $\omega_b(c(x))\geq p^k\big(b+i+(b-1)-(b-1)\big)=p^k(b+i)$. 
Assume that $\omega_H(g(x))=i+j$ for some  $2-i\leq j\leq b-1$.
 It is easy to see that $\omega_H(g(x))+z=d+1$ where $z=|\{l\mid 0\leq l\leq d \mbox{ and } g_l=0\}|$.
 We claim that, $z\geq b-j-\zeta$. Otherwise $d+1<\omega_H(g(x))+b-j-\zeta=i+j+b-j-(p^{e-k}-d-1)=i+b-p^{e-k}+d+1$.
 But $b+i\leq p^{e-k}$ leads us to a contradiction. Therefore the claim holds. So, 
 $\omega_b(g(x))\geq i+j+b-j-\zeta+(b-1)$.
Thus $\omega_b(c(x))\geq p^k\big(\omega_b(g(x))-(b-1)+\zeta\big)=p^k(i+b)$.
Hence $d_b(\mathcal{C}_{p^e-p^{e-k}+i})\geq p^k(i+b)$. Moreover, by part (1) of Lemma \ref{basiclemma3},
$\omega_b((x-1)^{p^e-p^{e-k}+i})=p^k\omega_b((x-1)^i)=p^k(b+i)$. Consequently, $d_b(\mathcal{C}_{p^e-p^{e-k}+i})=p^k(b+i)$. 
\end{proof}

\vspace{5mm} \noindent \footnotesize 
\begin{minipage}[b]{10cm}
Hojjat Mostafanasab \\
Department of Mathematics and Applications, \\ 
University of Mohaghegh Ardabili, \\ 
P. O. Box 179, Ardabil, Iran. \\
Email: h.mostafanasab@gmail.com, \hspace{1mm} h.mostafanasab@uma.ac.ir
\end{minipage}\\

\vspace{5mm} \noindent \footnotesize 
\begin{minipage}[b]{10cm}
Esra Sengelen Sevim \\
Eski Silahtara\v{g}a Elektrik Santrali, Kazim Karabekir,\\
Istanbul Bilgi University,\\
Cad. No: 2/1334060, Ey\"{u}p Istanbul, Turkey.\\ 
Email: esra.sengelen@bilgi.edu.tr
\end{minipage}\\

\end{document}